\documentclass[11pt]{article}

\usepackage{fullpage}
\usepackage{amsmath,amsfonts,amsthm,mathrsfs,mathpazo,xspace,hyperref,graphicx}
\usepackage{endnotes,needspace}
\usepackage{color}
\usepackage{bm}
\usepackage{times}
\usepackage{enumerate}
\usepackage{amssymb,latexsym}

\newenvironment{topbotframe} 
   {\needspace{2\baselineskip}\noindent\hrulefill\vspace{-0.1cm}}
   {\vskip -0.3cm\noindent\hrulefill\medskip\needspace{-2\baselineskip}}

\newenvironment{algorithm*}[1]
  {
    \begin{center}
      \hrulefill\\
      \textbf{#1}
  }
  {
    \vspace{-1\baselineskip}
    \hrulefill
    \end{center}
  }

\newcommand{\ket}[1]{|#1\rangle}
\newcommand{\bra}[1]{\langle#1|}

\newcommand{\ip}[1]{\langle {#1} \rangle}

\newcommand{\norm}[1]{\| #1 \|}

\DeclareMathOperator{\tr}{tr}

\newcommand{\beq}{\begin{equation}}
\newcommand{\eeq}{\end{equation}}
\newcommand{\beqn}{\begin{equation*}}
\newcommand{\eeqn}{\end{equation*}}

\newcommand{\C}{\ensuremath{\mathbb{C}}}

\newcommand{\mH}{\ensuremath{\mathcal{H}}}

\newcommand{\ce}{\ensuremath{c_{\eps}}}

\newtheorem{theorem}{Theorem}[section]
\newtheorem{thm}{Theorem}

\newtheorem{lemma}[theorem]{Lemma}

\newtheorem{claim}[theorem]{Claim}

\newtheorem{corollary}[theorem]{Corollary}
\newtheorem{definition}[theorem]{Definition}

\newtheoremstyle{promise}
{3pt}
{3pt}
{}
{}
{\itshape}
{:}
{.5em}
{}

\theoremstyle{promise}

\newcommand{\be}{\begin{eqnarray}}
\newcommand{\ee}{\end{eqnarray}}

\newcommand{\Id}{\ensuremath{\mathop{\rm Id}\nolimits}}

\newcommand{\eps}{\varepsilon}

\newcommand{\trim}{\mathrm{trim}}
\newcommand{\ls}{\mathrm{ls}}
\newcommand{\cont}{\mathrm{cont}}

\newcommand{\crd}{\color{red}}

\newcommand{\ignore}[1]{}

\begin{document}

\title{A polynomial-time algorithm for the ground state of $1$D gapped local Hamiltonians
}
\author{Zeph Landau\thanks{Computer Science Division, University of California, Berkeley. Supported by ARO Grant W911NF-12-1-0541, NSF Grant CCF-0905626 and Templeton Foundation Grant 21674. } \qquad Umesh Vazirani\footnotemark[1] \qquad Thomas Vidick\thanks{Computer Science and Artificial Intelligence Laboratory, Massachusetts Institute of Technology. Part of this work was completed while the author was visiting UC Berkeley. Supported by the National Science Foundation under Grant No. 0844626 and by the Ministry of Education, Singapore under the Tier 3 grant MOE2012-T3-1-009.}}
\maketitle

\begin{abstract}
Computing ground states of local Hamiltonians is a fundamental problem in condensed matter physics.  We give the first randomized polynomial-time algorithm for finding ground states of gapped one-dimensional Hamiltonians: it outputs an (inverse-polynomial) approximation, expressed as a matrix product state (MPS) of polynomial bond dimension.  The algorithm combines many ingredients, including recently discovered  structural features of gapped 1D systems, convex programming, insights from classical algorithms for 1D satisfiability, and new techniques for manipulating and bounding the complexity of MPS. Our result provides one of the first major classes of Hamiltonians for which computing ground states is provably tractable despite the exponential nature of the objects involved.


\end{abstract}


The exponential nature of quantum systems --- e.g. that $\text{exp}(n)$ parameters are necessary to specify the state of an $n$-particle quantum system --- is a double edged sword: while making quantum computers possible it is also an enormous obstacle to analyzing and understanding physical systems. In particular, this severely constrains the prospects for simulating quantum systems on a classical computer, which was the main issue raised in Feynman's seminal paper~\cite{Feynman82simulating}.  One is left to wonder which, if any, quantum systems occurring in nature are not as complex as it appears, and can be efficiently simulated?

Quantum complexity theory provides cause for pessimism --- it shows that placing restrictions solely on the type of local interactions is not a fruitful path to computational tractability, since instances with even the simplest two-body interactions lead to computational universality. Restricting the topology does not appear to help either --- approximating the energy of ground states of 1D local Hamiltonians is QMA-hard~\cite{AharonovGIK091d}, and even placing the further restriction of translation invariance does not help~\cite{GottesmanI13translation}. On the flip side, one of the rare bright spots for efficient simulation is the heuristic density matrix renormalization group (DMRG) algorithm~\cite{White92dmrg}, which has been extremely successful in practice in classically simulating 1D quantum many-body systems, since its introduction almost twenty years ago. Is there a way to reconcile the negative results from quantum complexity theory with the practical success of DMRG? This leads to the following refinement of the question above: ``Is there a natural subclass of 1D Hamiltonians for which the problem of approximating the ground state can be solved efficiently on a classical computer?''. Of course, such an approximation must support efficient classical computation of local observables such as energy. In other words, do quantum states for 1D systems (and possibly even more general systems) occurring in nature live in a small corner of Hilbert space \emph{that supports efficient classical computation}?

A natural subclass to consider are gapped 1D Hamiltonians (with constant spectral gap; see below). Unfortunately, folklore suggested that restricting to gapped 1D Hamiltonians did not help, a viewpoint reinforced by results from~\cite{SchuchC10algo} showing that a closely related problem is NP-hard. The recent sub-exponential algorithm given in~\cite{AradKLV12area}, which combined new structural results with an earlier algorithm based on dynamic programming~\cite{AharonovAI10algo}, strongly called this folklore into question. In this paper we  resolve this question by giving the first provably polynomial-time algorithm for approximating the ground state of 1D gapped quantum systems.

Formally, consider the Hamiltonian $H = \sum_{i=1}^{n-1} H_{i}$ acting on $n$ $d$-dimensional qudits, numbered from $1$ to $n$, where  $H_i$ is a Hermitian operator such that $0 \leq H_i \leq \Id$ acting on qudits $\{i,i+1\}$.  We shall assume there is a constant gap $\eps:= \eps _1 - \eps_0$ between the energy $\eps_0$ of the ground state $\ket{\Gamma}$ and the energy $\eps_1$ of the first excited state. Our goal is to efficiently find a succinct description of an approximation to $\ket{\Gamma}$ from which one can compute useful properties of the system; matrix product states (MPS) are known to provide such descriptions~\cite{Vidal03mps}.  

Our starting point is a structural result from Hastings' proof of an area law for 1D systems \cite{Hastings07area}. The result states that there exists an inverse polynomially close approximation to $\ket{\Gamma}$ with the property that its Schmidt rank across any cut of the 1D system is bounded by a polynomial in $n$.  This property already implies the existence of an efficient description of the ground state as an MPS, thus showing that the problem of finding such an approximation is in NP. Our algorithm finds such a succinct MPS approximation to $\ket{\Gamma}$ in polynomial time. For this, it combines many ingredients, including recently discovered  structural features of gapped 1D systems, convex programming, insights from classical algorithms for 1D satisfiability, and new techniques for manipulating and bounding the complexity of MPS. 

\begin{thm}\label{thm:main} Let $H$ be a $n$-qudit $1D$ local Hamiltonian with ground state $\ket{\Gamma}$ and gap $\eps$, and $\eta>0$. There is an algorithm that runs in time $n^{c(d, \epsilon)} poly(n/\eta)$, where $c(d,\epsilon) = 2^{O(\log^3 d/\epsilon)}$,
and with probability at least $1-1/poly(n)$ returns a matrix product state  representing a state $\ket{\Psi}$ such that $|\bra{\Psi}\Gamma\rangle|\geq 1-\eta$.
\end{thm}

The overall strategy of the algorithm is as follows.  In general, finding a minimum energy state can be expressed as a convex optimization problem:
$\min~\tr(H \sigma )$ subject to $\tr(\sigma) = 1$,~$\sigma \geq 0$, where $\sigma$ is a $n$-qubit density matrix describing the state. Unfortunately the dimension of the space on which $\sigma$ lives is exponentially large, making it hugely inefficient to solve the minimization directly. We observe that any efficient classical algorithm that relies on linear algebra must have the property that at any iteration it restricts itself to a calculation within a polynomial dimensional subspace, and furthermore that within that subspace, vectors are represented succinctly in a way that linear algebra can be performed efficiently.  Our algorithm progressively constructs a basis for such a subspace of polynomial dimension, guaranteed to contain a suitable approximation to $\ket{\Gamma}$.  Restricting the convex optimization problem to this subspace allows it to be solved efficiently.  Next we describe the outline of our algorithm for constructing the subspace and some of its ingredients.

We consider successive cuts stepping from left to right, with the $i$-th cut separating the the first $i$ qudits from the remaining $n-i$.  In the $i$-th iteration we create a subspace of polynomial dimension supported on the first $i$ qudits that contains the left Schmidt vectors of a good approximation to the ground state.  The subspace is specified by a spanning set of succinctly represented vectors, which we call a viable set. Upon completion we meet our goal since the subspace spanned by the final viable set is guaranteed to contain a good approximation to the ground state.

Na\"ively, since each iteration requires taking into account one further qudit it should result in an increase of the dimension of the subspace (size of the viable set) by a factor $d$, resulting in exponential growth. To overcome this, the left to right sweep of the algorithm is designed to exploit a key structural property of 1D systems, {\it approximate decoupling}. This property relies on the notion of {\it boundary contraction} of a state, which captures information representing how the left and right halves of the state are combined together. Subject to a fixed boundary contraction, the problem of finding a state of minimal energy can be {\it decoupled} into two disjoint problems --- to the left and to the right of the cut. As a consequence it is sufficient to guarantee that the current subspace contains, for each boundary contraction (taken from a suitably discretized $\epsilon$-net), the left Schmidt vectors of a good approximation to the ground state which has that boundary contraction across the cut. Unfortunately, this approach (which is a reformulation of the technique used in the exponential-time algorithm of~\cite{AharonovAI10algo,SchuchC10algo}) encounters a major difficulty: 
while it follows from the 1D area law~\cite{Hastings07area} that the dimension of the space of boundary contractions, which is proportional to the bond dimension across the cut, is at most polynomial in $n$, the size of the $\epsilon$-net is necessarily exponential in this dimension. So this does not yield an efficient algorithm.

To overcome this difficulty we appeal to two additional ideas. 1) A structural property that follows from existing proofs of the 1D area law: for any given cut and constant $\delta$ there exists a state having a \emph{constant} bond dimension $B_{\delta}$ (depending also on $d$ and $\eps$) across that cut (and polynomial across all others) that is a $\delta$-approximation to the ground state $\ket{\Gamma}$ (see Lemma~\ref{constantbondapprox} for a precise statement). 2) The use of a special operator $K$, known as an Approximate Ground State Projection (AGSP), which when applied to a state $\ket{\psi}$ improves its overlap with the ground state from $\delta$ to inverse polynomial in $n$, while increasing the bond dimension by only a polynomial factor. (We refer to Section~\ref{sec:agsp} for more details as well as a description of the specific AGSP that we use.)

These two ideas suggest the following modification to the above outline for an algorithm: proceed though the 1D chain from left to right, and at each iteration extend the viable set using a ($\ce/n$)-net over boundary contractions (for a suitably small constant $\ce$) of \emph{constant} bond dimension $B_{\ce}$. As a result, obtain a viable set supporting the left Schmidt vectors of a constant approximation to the ground state. Now apply  the AGSP $K$ to decrease the error to $\ce/n$ (so that the error is still small after $n$ iterations), and then iterate. 

This outline still presents some difficulties. In particular, how do we apply an AGSP to an unknown state, when we only have access to a subspace on which its ``left half'' is supported? 
Thinking through this reveals that the only meaningful way to apply the AGSP involves decomposing it into a polynomial number of terms, each of the form $A_j \otimes B_j$ across the cut.  Unfortunately, applying the individual $A_j$'s to vectors spanning the initial subspace results in a polynomial factor blow-up in its size. This is where the property of approximate decoupling is used: it lets us identify, among the large subspace spanned by all vectors thus obtained, a smaller-dimensional one which is still guaranteed to support a good approximation to the ground state. In order to present a more detailed description of the algorithm, we first give a precise definition of a viable set:

%

\begin{definition}\label{def:viable}
Given  $\delta>0$ and an integer $i$, $1\leq i\leq n$,  a set $S \subseteq (\C^d)^{\otimes i}$ is said to be $(i, \delta)$-viable for $\ket{\Gamma}$ if there exists  a state $\ket{\phi}\in (\C^d)^{\otimes n}$ such that $|\bra{\phi}\Gamma\rangle|\geq 1-\delta$ and such that the reduced density of $\ket{\phi}$ on the first $i$ qudits is supported on $\mathrm{Span}(S)$; we shall call such a state $\ket{\phi}$ a {\em witness} for $S$ and $\delta$ the \emph{error} of $S$. 

We will further say that the set $S$ is $(i, s,b, \delta)$-viable for $\ket{\Gamma}$ if $|S|\leq s$ and each $v\in S$ can be described by an MPS with maximum bond dimension at most $b$.  
\end{definition} 

For fixed polynomials $p(n), p_1(n), p_2(n)$ and some constant $\ce$ depending only on $\eps$, our algorithm constructs a $(i, p(n)p_1(n), p(n)p_2(n), \frac{\ce}{n})$-viable set from left to right, extending by one qudit in each iteration.  As stated above, a good approximation to the ground state can be computed from the last viable set (for $i=n$) by solving a simple convex optimization problem of polynomial size. 

The process of extending a $(i-1,\delta)$-viable set $S_{i-1}$ one qudit to the right to a $(i,\delta)$-viable set $S_{i}$ involves four steps, associated with the four parameters of a viable set. The first step increases $i-1$ to $i$, resulting in a multiplicative increase in the size $s$ of the set. The remaining three steps replace the value of each one of the other parameters (size, bond dimension and error) by a fixed polynomial in $n$. Together with the fact that no parameter gets blown up by more than a polynomial factor while carrying out any step implies that the parameters of the viable set at the end of the algorithm remain polynomially bounded. We outline these steps below:


\begin{enumerate}
\item {\bf Extension.} We start by extending the viable set $S_{i-1}$ for $\ket{\Gamma}$ to a viable set $S^{(1)}_{i} :=S_{i-1} \otimes \C ^d$. This results in a multiplicative factor $d$ increase in the size of the viable set.
\item {\bf Cardinality reduction.} Fix a $(\ce/n)$-net over the space of boundary contractions of states  with constant bond dimension $B_{\ce}$ across the $(i, i+1)$ cut; such a net has polynomial size. Using the principle of approximate decoupling and the fact that $S_{i}^{(1)}$ is a $(i,\ce/n)$-viable set we have the guarantee that there exists a net element that can be combined with a \emph{constant} number of vectors supported on $S_{i}^{(1)}$ to form the ``left half'' of a constant approximation to the ground state. As we will see, these vectors can be found by solving a polynomial-size convex optimization procedure. Considering all contractions in the net, the result is a viable set $S^{(2)}_{i}$ of fixed polynomial cardinality. Unfortunately, both the error and the bond dimension of vectors in the set have now blown up. 
\item {\bf Bond trimming.} Construct $S^{(3)}_{i}$ by truncating the bonds of all the elements in $S^{(2)}_{i}$ to some fixed polynomial. This results in a small increase in the error. 
\item {\bf Error Reduction.} Apply an AGSP to all vectors in $S^{(3)}_{i}$, resulting in a viable set $S_{i}$ with improved error, and size and bond dimension multiplied by a fixed polynomial. 
\end{enumerate}

The following chart tracks the changing of the four parameters during each step, highlighting when parameters 
are reset (red) by a fixed polynomial in $n$.  It is this resetting that ensures that the parameters do not blow up over the iterations as $i$ increases from $1$ to $n$.

\begin{center}
\begin{tabular}{llllll}
&&$i$ & $s$ & $B$ & $\delta$   \\ \hline
Start & & $i-1$ & $p(n)p_1(n)$ &$p(n)p_2(n)$& $\ce/n$ \\
&&&&&\\
{\bf Extension:} & $\rightarrow$ & ${\crd i} $ & $ {dp(n)p_1(n)}$ &$p(n)p_2(n)$ & $\ce/n$ \\
{\bf Size Trimming:} & $\rightarrow$ & $i$ & ${\color{red} p_1(n)}$ &  $p'(n)p_2(n)$&  ${ 1/12}$  \\
 {\bf Bond Trimming: }    & $\rightarrow$& $i$ & $p_1(n)$& ${\color{red} p_2(n)}$& ${ 1/2}$ \\
  {\bf Error reduction:} & $\rightarrow$&  $ i$ & $p(n) p_1(n)$& $ p(n) p_2(n)$&  ${\color{red} \ce/n }$\\ 
\end{tabular}
\end{center}

\paragraph{Organization.} We start with some preliminaries in Section~\ref{sec:prelim}. Section~\ref{sec:algo} contains a detailed presentation and analysis of the algorithm, including a section devoted to each of the four steps. We give some concluding remarks in Section~\ref{sec:conclusion}.

\section{Notation and Preliminaries}\label{sec:prelim}

Throughout, we consider a Hamiltonian $H = \sum_{i=1}^{n-1} H_{i}$ acting on $n$ $d$-dimensional qudits, indexed $1, \dots, n$ from left to right. Here  $H_i$  acts on qudits $\{i,i+1\}$ and satisfies $0 \leq H_i \leq \Id$.  We shall assume there is a constant gap $\eps:= \eps _1 - \eps_0$ between the energy $\eps_0$ of the ground state $\ket{\Gamma}$ and the energy $\eps_1$ of the first excited state.

The canonical basis of $\C^d$ is denoted by $\{\ket{1},\ldots,\ket{d}\}$. We write $\mH$ for the Hilbert space $(\C ^d)^{\otimes n}$ corresponding to the $n$ qudits and $\mH_{[i,j]}$ for the Hilbert space of the subset  of  qudits with indices in $[i,j]$; we also write $\mH_i$ for $\mH_{[i,i]}$.  For any density matrix $\rho$ acting on $\mH$, $tr_{[i,j]} \rho$ will denote the tracing out of the qudits with indices in $[i,j]$.

 Throughout the algorithm vectors in $(\C^d)^{\otimes n}$ will be represented as matrix product states (MPS), which can be specified as a sequence of tensors $A_1,\ldots,A_n$ where $A_1\in \C^d\times \C^{B_1}$, $A_i\in \C^{B_{i-1}}\times\C^d\times\C^{B_{i}}$ for $1\leq i < n$, and $A_n\in \C^{B_{n-1}}\otimes \C^d$. We will refer to $B_i$ as the \emph{bond dimension} across cut $(i,i+1)$. For more on MPS we refer to~\cite{Perez07MPS}.

The constant $\ce:= (\frac{\eps}{169})^2$ will play a particular role in our analysis. We note that it satisfies the following inequalities:
\begin{equation}\label{eq:def-ce}
\ce ( 1 + \frac{1}{\eps}) \leq \frac{1}{2},\qquad  14 \sqrt{\ce}/\eps <\frac{1}{12}, \qquad \frac{84 \ce}{\eps} < \frac{1}{2}.
\end{equation} 

We state two simple lemmas that will be used repeatedly in our analysis. 

\begin{lemma}\label{claim:energy-overlap}
Suppose a state $\ket{v}$ has energy $\bra{v}H\ket{v}\leq \eps_0+\delta$, for some $0\leq \delta \leq \eps$. Then $|\bra{v}\Gamma\rangle|\geq 1-\delta/\eps$. 
\end{lemma}

\begin{proof}
Write $\ket{v} = \lambda\ket{\Gamma} + \sqrt{1-\lambda^2}\ket{\Gamma^\perp}$ for some unit vector $\ket{\Gamma^\perp}$ orthogonal to $\ket{\Gamma}$. $\ket{v}$ has energy
$$ \eps_0 + \delta \,\geq\, \bra{v} H \ket{v} \,\geq \,\lambda^2 \eps_0 + (1-\lambda^2)\, \eps_1,$$
which gives $\lambda^2 \geq 1-\delta/\eps$, hence $\lambda \geq 1-\delta/\eps$. 
\end{proof}

\begin{lemma}\label{claim:overlap}
Let $0\leq \delta,\delta'\leq 1$ and $\ket{v}$, $\ket{v'}$ and $\ket{w}$ be states such that $|\bra{v}w\rangle|\geq 1-\delta$ and $|\bra{v'}w\rangle|\geq 1-\delta'$. Then $|\bra{v}v'\rangle|\geq 1-2(\delta+\delta')$.
\end{lemma}

\begin{proof} We have
\begin{align*}
|\bra{v}v'\rangle| &\geq |\bra{v}\Gamma\rangle\bra{v'}\Gamma\rangle| - \big((1-|\bra{v}\Gamma\rangle|^2)(1-|\bra{v'}\Gamma\rangle|^2)\big)^{1/2}\\
&= (1-\delta)(1-\delta') - 2\sqrt{\delta\delta'}\\
&\geq 1-2(\delta+\delta').
\end{align*}
\end{proof}

\subsection{Low entanglement approximations of the ground state}

\begin{definition}  Given a vector $\ket{v}\in\mH$, by a {\it Schmidt decomposition across the $(i,i+1)$ cut} we shall mean a decomposition $\ket{v}=\sum_{j=1}^D \lambda_j \ket{a_j}\ket{b_j}$ with $\{\ket{a_j}\}$ (respectively $\{\ket{b_j}\}$ ) a family of orthonormal vectors of $\mH_{[1,i]}$ (respectively $\mH_{[i+1,n]}$) and with $\lambda_j \geq \lambda_{j+1} >0$ for all $1\leq j\leq D$. The vectors $\ket{a_j}$ will be called the \emph{left Schmidt vectors} across that cut, and the vectors $\ket{b_j}$ the \emph{right Schmidt vectors}; $D$ is the \emph{Schmidt rank} across the cut. 
\end{definition}

The following lemma follows from the 1D area law~\cite{Hastings07area}. Although we will only need a polynomial bound on the Schmidt rank, we state the lemma using the best known parameters~\cite[Section 7]{AradKLV12area}.

\begin{lemma}  \label{l:mpsapprox} For any constant $c>0$ there is a constant $C\geq 1$ such that for every $n$ there is a vector $\ket{v}$ with Schmidt rank bounded by $exp(C(\ln n)^{3/4}\eps^{-1/4})$ across every cut such that $|\ip{\Gamma | v}|\geq 1 - n^{-c}$. 
\end{lemma}

The operation of \emph{trimming} a state across a cut --- removing Schmidt vectors associated with the smallest Schmidt coefficients --- will be used repeatedly by our algorithm. 

\begin{definition}  Given a state $\ket{v}\in\mH$ with Schmidt decomposition $\ket{v}=\sum_j \lambda_j \ket{a_j}\ket{b_j}$ across the $(i,i+1)$ cut and an integer $D$, define $\trim^i_{D} \ket{v} := \sum_{j=1}^D \lambda_j \ket{a_j}\ket{b_j}$.\footnote{We note an ambiguity in the definition of $\trim^i_{D} \ket{v}$ in the case of degeneracies among the Schmidt decomposition. In our analysis it will never matter which eigenvectors associated with the same eigenvalue are kept.}
\end{definition}

The following well-known lemma states that among all vectors with Schmidt rank $D$ across a certain cut $i$, $\trim^i_{D} \ket{v}$ provides the closest approximation to $\ket{v}$. 

\begin{lemma}[Eckart-Young theorem]\label{lem:eckart}
 Let $\ket{v}\in \mH$ have Schmidt decomposition $\ket{v}=\sum_i \lambda_i \ket{a_i}\ket{v_i}$ across the $(i,i+1)$ cut. Then for any integer $D$ the vector $\ket{v'}=\trim^i_{D} \ket{v}/\norm{\trim^i_{D} \ket{v}}$ is such that $\bra{v'}v\rangle \geq |\bra{w}v\rangle|$ for any unit $\ket{w}$ of Schmidt rank at most $D$ across the $i$-th cut. 
\end{lemma}

We will require the existence of close approximations to the ground state that have \emph{constant} Schmidt rank across a given cut (and polynomial across the others). 

\begin{lemma}[\cite{Hastings07area,AradKLV12area}] \label{constantbondapprox} For any cut $(i,i+1)$ and any constant $\delta$, there exists a constant $B_{\delta} = exp(O(1/\eps \log^3 d \log 1/\delta))$ such that the state $\ket{v} := \trim^i_{B_{\delta}}\ket{\Gamma}/ \norm{\trim^i_{B_{\delta}}\ket{\Gamma}}$ has the property that $|\ip{\Gamma|v}|\geq 1- \delta$.
\end{lemma}

\begin{lemma}  \label{l:2} Let $\delta>0$ be such that $\delta(1+ 1/\eps)\leq \frac{1}{2}$ and $\ket{w}$ a vector with energy no larger than $\eps_0 + \delta$. Then $\ket{v}:=\trim_{B_{\delta}}\ket{w} / \norm{\trim_{B_{\delta}}\ket{w} }$ has energy no larger than $\eps_0 + 6 \sqrt{\delta}$.
\end{lemma}

\begin{proof}  By Lemma \ref{claim:energy-overlap}, $|\ip{\Gamma | w}| \geq 1 - \delta/\eps$. Let $\ket{u} := \trim^i_{B_{\delta}}\ket{\Gamma}/ \norm{\trim^i_{B_{\delta}}\ket{\Gamma}}$. Since by Lemma \ref{constantbondapprox}, $|\ip{\Gamma|u}|\geq 1- \delta$, Lemma \ref{claim:overlap} implies $|\ip{w| u}| \geq 1 - \delta(1+ 1/\eps)= 1- \delta'$.  The Eckart-Young theorem (see Lemma~\ref{lem:eckart}) 
therefore implies that $\norm{ \trim^i_{B_\delta} \ket{w}}^2 \geq 1 -  \delta'$.  Set $\ket{w_0}= \trim^i_{B_\delta} \ket{w}$, $\ket{w_1} = \ket{w} - \ket{w_0}$. We have
\[ \ip{w|H|w}= \ip{w_0|H-H_i|w_0} + \ip{w_1|H-H_i|w_1} +\ip{ w |H_i|w}, \] and it follows that
\[ \ip{w_0|H|w_0} \leq \eps_0 + \delta -  \ip{w_1|H-H_i|w_1} + \ip{w_0|H_i|w_0} - \ip{ w |H_i|w}. \]
Using the fact that $|\ip{w_0|H_i|w_0} - \ip{ w |H_i|w}| = | \ip{w_0|H_i|w_0-w} +\ip{w_0-w|H_i|w}| \leq 2\norm{\ket{w}- \ket{w_0}} \leq 2 \sqrt{\delta'} $ along with the lower bound of   $\ip{w_1|H-H_i|w_1} \geq (\eps_0 -1) \norm{\ket{w_1}}^2 \geq (\eps_0 -1) \delta'$ we have:
\[  \ip{w_0|H|w_0} \leq \eps_0 (1 - \delta') + \delta + 2 \sqrt{\delta}. \]
This implies:
\[ \ip{v|H|v}  \leq \eps_0 + \frac{ \delta + 2 \sqrt{\delta}}{1 - \delta'} \leq \eps_0 + 6 \sqrt{\delta},\] if $\delta' \leq1/2$. 

\end{proof}

\begin{corollary} \label{c:1}For any cut $(i,i+1)$ and any constant $\delta$, there exists a constant $B_{\delta}$ such that there exists a  state $\ket{v}$ with Schmidt rank $B_{\delta}$ that has energy at most $\eps_0 + 6 \sqrt{\delta}$ as well as the property that $\ip{\Gamma|v}\geq 1- \delta$.
\end{corollary}
\begin{proof}  Setting $\ket{v}= \trim^i_{B_{\delta}}\ket{\Gamma}/ \norm{\trim^i_{B_{\delta}}\ket{\Gamma}}$, the result follows from Lemma~\ref{constantbondapprox} and Lemma~\ref{l:2}.
\end{proof}

\section{Algorithm and analysis}\label{sec:algo}

As sketched in the introduction, in the $i$-th iteration our algorithm constructs a $(i, s, b, c_\eps/n)$-viable set, where $
\ce$ is a constant depending only on $\eps$ that satisfies~\eqref{eq:def-ce}. The four steps in each iteration are designed to update the four parameters of the viable set. As we will show, these updates always satisfy the condition that the parameters $s$ and $b$ are each bounded by some fixed polynomial in $n$ of degree independent of the iteration $i$. 

The initialization step $i=0$ is trivial, as the set $\{1\}$ is a $(0, \delta)$-viable set for any $\delta\geq 0$, $b\geq 0$ and $s\geq 1$.  Let $1\leq i\leq n$ be an integer, and $S_{i-1}$ the $(i-1, s, b, \ce/n)$-viable set obtained at the end of the $(i-1)$-st iteration of the algorithm, where $s$ and $b$ are both polynomial in $n$. 
In the following four subsections we describe in detail how each of the four steps of the algorithm can be performed efficiently, and track the changes in the parameters of the viable set. 

\subsection{Extension}

The first step in the $i$-th iteration involves extending the set $S_{i-1}$ to an $(i, ds, b, \ce/n)$-viable set $S_{i}^{(1)}$ as follows: 
\bigskip

\begin{topbotframe}
\begin{quote} {\bf Algorithm step 1: extension}\\
 Let $S_{i-1}$ be a $(i-1, s, b, \ce/n)$-viable set.
\begin{enumerate}
\item[1.] Return $S_{i}^{(1)} \,:=\, \big\{ \ket{s} \ket{j}: \,\ket{s}\in S_{i-1},\, \ 1\leq j \leq d\big\}.$
\end{enumerate}
\end{quote}
\end{topbotframe}

\bigskip

The computation of $S_{i}^{(1)}$ from $S_{i-1}$ can clearly be done efficiently: MPS representations for vectors in the latter are constructed as the concatenation of an MPS for a vector in the former with an independent tensor corresponding to the additional basis state $\ket{j}$. The following claim shows that $S_{i}^{(1)}$ has the required properties. 

\begin{claim} \label{c:ext}
$S_{i}^{(1)}$ is an $(i, ds, b, \ce/n)$-viable set.
\end{claim}

\begin{proof} By definition $|S_{i}^{(1)}| = d|S_{i-1}|$. Clearly, the bond dimensions of vectors in $S_{i}^{(1)}$ are no larger than that of vectors in $S_{i-1}$. Given a witness $\ket{v}$ for $S_{i-1}$ with Schmidt decomposition across the $(i-1,i)$ cut $\ket{v} = \sum_j \lambda_j \ket{s_j} \ket{t_j}$, decompose the first qudit of $\ket{t_j}$ on the computational basis as $\ket{t_j}= \sum_{k=0}^{d-1} \ket{k} \ket{t_{jk}}$.  Then clearly $\ket{v}= \sum_{j,k} \lambda_j \ket{s_j}\ket{k} \ket{t_{jk}}$ is also a witness for $S_{i}^{(1)}$.
\end{proof}

\subsection{Cardinality reduction}

This is the main step in which the 1D nature of the Hamiltonian is exploited. It is instructive to think by analogy about a classical 1D constraint satisfaction problem (CSP), where assigning a value to one of the variables decomposes the problem into two parts, to the left and the right. By contrast, specifying the density matrix of one of the qudits does not decompose the problem of computing the ground state into left and right halves. After all, the qudit's density matrix is $d^2$-dimensional, whereas for a generic state of the system the Schmidt rank between the left and right halves could grow with $n$, the total number of qudits. The correct notion for carrying out such a decomposition in the quantum case is that of a boundary contraction, which may be thought of as a density matrix on a qudit \emph{and} the bond representing the entanglement between the left and right halves of the state. Such a decomposition will in general no longer be exact, and there is a tradeoff between the error introduced and the bond dimension. To make these concepts precise, in the following subsection we introduce the critical notions of left state and boundary contraction.

\subsubsection{Boundary contractions}

Given $1\leq i \leq n$, we will write $\mH_L$ for the space $\mathcal{H} _{[1,i]}$ and $\mH_R$ for the space $\mathcal{H} _{[i+1,n]}$. We also let $H_L = H_1+\cdots+H_{i-1}$ and $H_R = H_{i+1}+\cdots + H_{n-1}$, so that the total Hamiltonian $H=H_L+H_i+H_R$, where $H_i$ is the only term acting across the $(i,i+1)$ cut. 
\begin{definition}  Given a state of Schmidt rank $B$ and Schmidt  decomposition across the $(i,i+1)$ cut given by  $\ket{v}=\sum_{j=1}^B \lambda_j \ket{a_j}\ket{b_j}$, let $U_{v}:\C ^B \rightarrow \mathcal{H} _{R}$ be the partial isometry specified by $U_{v} \ket{j}= \ket{b_j}$.  By abuse of notation we also write $U_v^*$ for $I \otimes U_v^*$ when acting on $\mH_{[k,n]}$ for $k \leq i+1$.
\begin{itemize}
\item Define the {\it left state} of $\ket{v}$ to be $\ls(v):= U^*_v \ket{v}= \sum_j \lambda_j \ket{a_j}\ket{j}\in\mH_L\otimes\C^B$. 
\item Define the {\it boundary contraction } of $v$ as 
$$\cont(v):=  \tr_{[1,\dots, i-1] }( \ket{\ls(v)} \bra{\ls(v)})= U_v \tr_{[1,\dots, i-1] }(\ket{v}\bra{v}) U_v^*.$$
Then $\cont(v)$ is a density matrix supported on $\mH_i\otimes \C^B$. 
\end{itemize}
\end{definition}

The following claim (specifically part 3) shows how the boundary contraction can be used to decompose the problem of finding an approximate ground state into independent ``left'' and ``right'' subproblems:

\begin{lemma}[Gluing] \label{l:gluing} Given a density matrix $\sigma$ on the space $\mH_L \otimes \C^B$ and a state $\ket{v}= \sum_{j=1}^B \lambda_j \ket{a_j}\ket{b_j}$ on $\mH_L \otimes \mH_R$ the density matrix $\sigma':= U_{v} \sigma U_{v}^*$ on $\mH_L \otimes \mH_R$ satisfies the following properties:
\begin{enumerate}
\item $\tr_{\mH_R} ( \sigma') = \tr_{\C^B} (\sigma),$
\item $\| \tr_{[1,\cdots, i-1]}  (\sigma') - \tr_{[1,\cdots, i-1]} (\ket{v}\bra{v})\|_1 = \|\tr_{[1, \dots, i-1]} ( \sigma) - \cont(v)\|_1,$
\item $\tr(\sigma' H) \leq \tr(\sigma H_L) + \tr(\ket{v}\bra{v} ( H_R + H_{i})) + n  \|\tr_{[1, \dots, i-1]} ( \sigma) - \cont(v)\|_1.$
\end{enumerate}
\end{lemma}

\begin{proof}
\begin{enumerate}
\item Clear, since $U_v$ is unitary. 
\item We have
\begin{align*}
\|\tr_{[1,\cdots, i-1]}  (\sigma') - \tr_{[1,\cdots, i-1]} (\ket{v}\bra{v})\|_1 &= 
\| \tr_{[1,\cdots, i-1]}  (\sigma' - \ket{v}\bra{v})\|_1 \\
&= \| \tr_{[1,\cdots, i-1]}  (U_v^* \sigma' U_v  - U_v^* \ket{v}\bra{v}) U_v\|_1 \\
&= \|\tr_{[1, \dots, i-1]} ( \sigma) - \cont(v)\|_1.
\end{align*}
\item Write $\tr(\sigma ' H) = \tr( \sigma ' H_L) + \tr (\sigma ' (H_i +H_R))$. By the first item, $\tr(\sigma' H_L) = \tr(\sigma H_L)$.
By the second item,
\begin{align*}
\tr(\sigma' (H_i + H_R)) - \tr(\ket{v}\bra{v}(H_i + H_R))
&= (\tr_{[1, \dots, i-1]}  \sigma - \cont(v))(H_i + H_R))\\
&\leq n \|\tr_{[1, \dots, i-1]} ( \sigma) - \cont(v)\|_1.
\end{align*}
\end{enumerate}
\end{proof}

We will make use of a $\eta$-net over the unit ball of boundary contractions for the trace norm, for some $\eta>0$ to be determined later. Such a net can be efficiently constructed by discretizing a region of $\C^B \otimes \C^d$ containing its unit ball, as shown in the following lemma. 

\begin{lemma}\label{l:constantbondapprox}
For any integers $B,d$ and $\eta>0$ let 
$$\mathcal{I}_\eta \,:=\, \{-1,-1+\lfloor\eta/(Bd)^2\rfloor,\ldots,1-\lfloor\eta/(Bd)^2\rfloor,1\}$$
 and $\mathcal{N}_\eta := \mathcal{I}_\eta^{B}\times\mathcal{I}_\eta^d$. Then 
the set $\mathcal{N}_\eta$ has cardinality at most $(2\lceil Bd/\eta\rceil+1)^{2Bd}$ and is such that for every $Y\in \C^B \otimes \C^d$ with trace norm at most $1$, there is an $X\in \mathcal{N}_\delta$ such that $\|Y-X\|_1\leq \eta$.
\end{lemma}

\begin{proof}
The bound on the cardinality of $\mathcal{N}_\eta$ is clear. For the distance, we note that $\|Y\|_1\leq 1$ implies that each entry of $Y$ has modulus at most $1$, and  bound $\|Y-X\|_1\leq (Bd)^2 \max_{i,j}|Y_{ij}-X_{ij}|$.
\end{proof}

\subsubsection{The size trimming convex program}\label{sec:size-convex}

We are now ready to describe the procedure for reducing the cardinality of a viable set.  Let $\mathcal{N}$ be the $(c_\eps/2n)$-net over the space of boundary contractions with constant bond dimension $B_{\ce}$ obtained from Lemma~\ref{l:constantbondapprox}.
For each  $X\in \mathcal{N}$, by solving a suitable convex program we find a state on $\mathcal{H}_L\otimes \C^{B_{\ce}}$ of minimum energy among those states whose boundary contraction (reduced density matrix on $\mH_i\otimes \C^{B_{\ce}}$)  is close to $X$. The new viable set is then the union over all elements of $\mathcal{N}$ of the left Schmidt vectors, on $\mathcal{H}_L$, of these states.  Unfortunately, the dimension of this convex program scales with the dimension of $\mathcal{H}_L$, which is exponential in $n$, and to solve it efficiently we must restrict the optimization to states supported on a subspace $S$ of polynomial dimension.  Before filling in details of this sketch, we describe the actual steps of the resulting algorithm:

\begin{topbotframe}
\begin{quote}{\bf Algorithm steps 2 and 3: size trimming and bond trimming}
 
 Let $S_{i}^{(1)}$ be the $(i, d s, b, \ce/n)$-viable set constructed as a result of the extension step described in the previous section.
\begin{enumerate}
\item[2.] For each $X\in\mathcal{N}$, solve the following \emph{size trimming convex program}, whose variable is a density matrix $\sigma$ supported on the space $\textrm{Span}\{S_{i}^{(1)}\} \otimes \C^{B_{\ce}} \subseteq \mH_L \otimes \C^{B_{\ce}}$:\begin{align}
\mathrm{min}&\quad \sum_{j=1}^{i-1}\, \tr(H_j \,\sigma )\label{eq:conv}\\
\text{such that}&\quad \big\|\tr_{[1, \ldots, i-1]}(\sigma) - X\|_1 \leq \frac{\ce}{2n},\notag\\
&\quad \tr(\sigma) = 1,\quad \sigma \geq 0.\notag
\end{align} 
Let $\ket{u}= \sum_j \ket{u_j}\ket{j}$ be the leading eigenvector of the solution $\sigma$ to this program, and let $S_{i}^{(2)}$ be the set containing the union of all $\{\ket{u_j}\}$, obtained for each net element $X$. 
\item[3.]  Trim each of the bonds $1,\ldots,i-1$ of each $\ket{u}\in S_{i}^{(2)}$ to $p_2(n)$, where $p_2(n)$ is a polynomial defined in Claim~\ref{c:trim} below. Include the MPS representation of all resulting vectors in $S_{i}^{(3)}$.
\end{enumerate}
\end{quote}
\end{topbotframe}

We note that since the set $S_{i}^{(1)}$ contains vectors specified using polynomial-size MPS, for any $X$ a polynomial-size representation for the optimal solution $\sigma$ to the convex program~\eqref{eq:conv} can be computed efficiently. For this, we first compute an orthonormal basis $\{\ket{f_k}\}$ for $\text{Span}\{S_{i}^{(1)}\}$. Vectors in this basis van be represented as linear combinations of vectors in $S_{i}^{(1)}$. The variables of the convex program will be the polynomially many coefficients of $\sigma$ on the $\ket{f_k}\bra{f_\ell}$; to express the objective function as a function of these variables it suffices to compute each $\bra{f_k}H_j\ket{f_\ell}$, which can be done efficiently by expanding the $\ket{f_k}$ on the vectors of $S_{i}^{(1)}$ and evaluating the resulting expression by using the MPS representations of the latter. The constraints can also be expressed as convex functions of the variables by pre-computing all $\bra{f_k}H_j \ket{f_\ell}$. The remaining steps rely on the singular value decomposition which can be performed efficiently as well. 

In the following two subsections we successively analyze the properties of the sets $S_{i}^{(2)}$ and $S_{i}^{(3)}$. The final outcome will be that $S_{i}^{(3)}$ is a $(i, p_1(n), p_2(n), 1/2)$-viable set,  where $p_1 (n) $ is defined in Claim~\ref{c:2} and $p_2 (n)$ is defined in Claim~\ref{c:trim}. 
 
\subsubsection{Size trimming}
 
We show that the set $S_{i}^{(2)}$ defined in the second step of the algorithm is a $(i, p_1(n), q(n)b, 1/12)$-viable set, for some polynomial $q(n)$. 
The key observation is that, conditioned on the existence of a state $\ket{w}$ in $\text{Span}\{S_{i}^{(1)}\}\otimes \mH_R$ having both low energy and low bond dimension, the solution $\sigma$ of the size trimming convex program~\eqref{eq:conv}  for an $X$ sufficiently close to the boundary contraction of $\ket{w}$ allows for the easy computation of the left Schmidt vectors of a good approximation to the ground state.  This is shown in the following lemma; the subsequent Lemma \ref{l:within} establishes the existence of $\ket{w}$. 

\begin{lemma}\label{l:leading}
Suppose there exists a state $\ket{w}$ in $\text{Span}\{S_{i}^{(1)}\}\otimes \mH_R$ of bond dimension $B_{\ce}$ having energy at most $\eps_0 + 6 \sqrt{\ce}$. 
Let $X$ be the element of the net $\mathcal{N}$ that is closest to $\cont(w)$ and let $\sigma$ be the solution to the size trimming convex program~\eqref{eq:conv}. Let $\ket{u}=\sum_j \ket{u_j}\ket{j}$ be the leading eigenvector of $\sigma$. Then there exist orthonormal vectors $\{\ket{b_j} \in \mH_R \}$  such that $\ket{u'}:=\sum_j \ket{u_j}\ket{b_j}$ has energy at most $\eps_0 +  \eps/12$ and $| \ip{ u' | \Gamma}| \geq 1- 1/12$.
\end{lemma}

 \begin{proof}[Proof of Lemma~\ref{l:leading}]  Apply Lemma \ref{l:gluing} to $\sigma$ and $\ket{w}$ to conclude that the energy of $\sigma'=U_{w} \sigma U_{w}^*$ can be upper bounded as follows
\begin{align}
 \tr(\sigma'H)  &\leq \tr(\sigma H_L) + \tr(\ket{w}\bra{w} ( H_R + H_{i})) + n  \|\tr_{[1, \dots, i-1]} ( \sigma) - \cont(w)\|_1 \notag\\
&\leq \eps_0 + 6 \sqrt{\ce} + \ce \notag\\
&\leq \eps_0 + 7 \sqrt{\ce},\label{eq:leading-1} 
\end{align}
where we used the optimality of $\sigma$ to bound  $\tr(\sigma H_L) \leq \tr(\ket{w}\bra{w} H_L)$; indeed, $\ls(w)$ itself is a feasible solution to~\eqref{eq:conv}. Let $\ket{v_j}$ be the eigenvectors of $\sigma'$, with corresponding eigenvalues $\lambda_1\geq \cdots\geq\lambda_{B_{\ce}}$.  From~\eqref{eq:leading-1} we get that  $\sum_{j\in J} \lambda_j ^2\geq 1/2$ where $J=\{ j: \tr(H \ket{v_j}\bra{v_j}) \leq \eps_0 + 14 \sqrt{\ce}\}$.  But since by~\eqref{eq:def-ce} $14 \sqrt{\ce} < \eps/12 <\eps/2$, any $\ket{v_j}$ with energy less than $\eps_0 + 14 \sqrt{\ce}$ must satisfy $|\ip{v_j | \Gamma}|^2 > 1/2$. Thus there can only be one such $\ket{v_j}=\ket{v_1}$, and $\lambda_1^2 > 1/2$.  Letting $\ket{u}:=U_w^*\ket{v_1}$, $\ket{u}$ is the leading eigenvector of $\sigma$ and has energy at most $\eps_0 + \eps/12$.  Applying Lemma \ref{claim:energy-overlap} to $\ket{u'}:=\ket{v_1}$ establishes $| \ip{ u' | \Gamma}| \geq 1- 1/12$.
\end{proof}

In order to apply the previous lemma, we need to establish its hypothesis: that there exists a vector $\ket{w}$ with small bond dimension and low energy that lies in $\text{Span}\{S_{i}^{(1)}\}\otimes \mH_R$.

\begin{lemma} \label{l:within}
There exists $\ket{w}$ in $\text{Span}\{S_{i}^{(1)}\}\otimes \mH_R$ with bond dimension $B_{\ce}$ and energy bounded by $ \eps_0 + 6\sqrt{\ce}$.
\end{lemma}
\begin{proof}
Let $\ket{v}$ be the witness for $S_{i}^{(1)}$ being a $(i,\ce/n)$-viable set. 
Since $\ket{v}$ is $(\ce/n)$-close to $\ket{\Gamma}$ and $H$ has norm at most $n$, its energy $\ket{v}$ is upper bounded by $\eps_0 + \ce$. Applying Lemma \ref{l:2} to $\ket{v}$ we get a state $\ket{w}:= \trim_{B_{\ce}}\ket{v} / \norm{\trim_{B_{\ce}}\ket{v}} $ with energy bounded by $\eps_0 + 6 \sqrt{\ce}$; moreover the left Schmidt vectors of $\ket{w}$ still lie in $\text{Span}\{S_{i}^{(1)}\}$.  
\end{proof}

We combine Lemmas~\ref{l:leading} and~\ref{l:within} to show that  $S_{i}^{(2)}$, the result of step 2. of the algorithm, is a $(i, \frac{1}{12})$-viable set of a fixed polynomial size and polynomial bond dimension.

\begin{claim} \label{c:2}
If $S_{i}^{(1)}$ is a $(i, ds, b, \frac{\ce}{n})$-viable set then $S_{i}^{(2)}$ is a $(i, p_1(n), dsb, \frac{1}{12})$-viable set with $p_1(n) =\, B_{\ce} (4\lceil B_{\ce}d n/\ce\rceil+1)^{2B_{\ce} d}$.
\end{claim}
\begin{proof}
Together, Lemmas~\ref{l:leading} and~\ref{l:within} establish that the vector $\ket{u'}$ from Lemma~\ref{l:leading} is a witness for $S_{i}^{(2)}$ with error $\frac{1}{12}$.  For every element of the net $\mathcal{N}$, step 2. of the algorithm generates at most $B_{\ce}$ vectors to be added to $S_{i}^{(2)}$ yielding a bound on the cardinality of $S_{i}^{(2)}$ of $p_1(n)= B_{\ce} |\mathcal{N}|$ with $|\mathcal{N}|$ given in Lemma~\ref{l:constantbondapprox}.  Finally since $S_{i}^{(2)}\subset \text{Span}\{S_{i}^{(1)}\}$ it is clear that each vector in $S_{i}^{(2)}$ has an MPS description with bond dimension bounded by the product of $|S_{i}^{(1)}|$ and the maximal bond dimension of the elements of $S_{i}^{(1)}$, i.e. $dsb$.
\end{proof}

\subsubsection{Bond Trimming}

In this section we analyze the result of step 3. of the algorithm, the bond trimming step (see Section~\ref{sec:size-convex} for a description of that step).
The key property we use to bound the error incurred in this step is that the state being trimmed is close to a state with low Schmidt rank. We first prove a general lemma showing that trimming such a state cannot result in a large error:

\begin{lemma}\label{l:trim1} Given a vector $\ket{v}$ with $D$ non-zero Schmidt coefficients across the $(i,i+1)$ cut, for any $\ket{u}$ it holds that 
\[ \big| \ip{\trim^i_{D/\eps} (u) |v}\big| \geq |\ip{u|v}| - \eps. \]
\end{lemma}
\begin{proof}  Denote by $\lambda_1 \geq \lambda_2 \geq \dots $ the Schmidt coefficients of $\ket{u}$. We proceed by contradiction: assume $\ket{w}= \ket{u}  -\ket{\trim_{D/\eps} (u)}$ has the property that $\ip{w|v} > \eps$. Since $\ket{v}$ has only $D$ non-zero Schmidt coefficients, by the Eckart-Young theorem (Lemma~\ref{lem:eckart}) this last condition implies that $\sum_{i=1}^D |\lambda_{\lceil D/ \eps \rceil +i}|^2 > \eps$.  Using that the Schmidt coefficients are decreasing, we get 
$$\sum_{j=1} ^{{\lceil D/ \eps \rceil } } \,\lambda_j^2 \,\geq\, \Big\lceil\frac{1}{\eps}\Big\rceil\,\sum_{j=1} ^{D } \,\lambda_{\lceil D/ \eps \rceil +j}^2 \,>\, 1,$$
a contradiction. 
\end{proof}

We also show that trimming a state across a given bond does not increase the Schmidt rank across any of the other bonds. 

\begin{lemma} \label{l:trim2}
For any integer $m$ the Schmidt rank of $\ket{\trim^i_m (u)}$ is no larger than the Schmidt rank of $\ket{u}$ across any cut $(j,j+1)$.
\end{lemma}
\begin{proof}  Without loss of generality, assume $j\leq i$.  Writing the Schmidt decomposition across cut $(i,i+1)$ as $\ket{u}= \sum \lambda_i \ket{\alpha_i} \ket{\beta_i}$ notice that
\[\ket{\trim^i_m (u)} = \Big( \Id \otimes \sum_{k=1}^{m} \ket{\beta_k}\bra{\beta_k}\Big) \ket{u}. \]
Since this operator only acts strictly to the right of the $(j,j+1)$ cut, it cannot increase the Schmidt rank across that cut.
\end{proof}

Based on the previous two lemmas, we can show that the set produced by the bond trimming step of the algorithm has the required properties. Let $r(n)$ be a polynomial such that there exists a vector $\ket{v}$ with Schmidt rank $r(n)$ such that $|\ip{v|\Gamma}|> 1- 1/48$ (as is shown to exist in Lemma~\ref{l:mpsapprox}).

\begin{claim} \label{c:trim} The set $S_{i}^{(3)}$ produced at the end of step 3. of the algorithm is a $(i,p_1(n), p_2(n),1/2)$-viable set, where $p_2(n):= 48n r(n)$.
\end{claim}

\begin{proof} 
Based on the analysis of the previous step of the algorithm, from Lemma~\ref{l:trim1} we know that there exists a witness $\ket{u'}$ for $S_{i}^{(2)}$ such that $|\bra{u'}\Gamma\rangle|\geq 1-1/12$; furthermore $\ket{u} = \ls(u')$ is a member of the set $S_{i}^{(1)}$. Using Lemma~\ref{claim:overlap} we get that $|\ip{v|u'}|\geq 1-5/24$. 
Applying Lemmas~\ref{l:trim1} and~\ref{l:trim2} to $\ket{u'}$ yields that the successive trimming of $\ket{u'}$ for each of the bonds $1, \dots, i$ to Schmidt rank $p_2(n)$ results in a state $\ket{w}$ with 
$$|\ip{v | w}| \geq |\ip{v|u'}| - n/(48n) \geq 1- 11/48.$$
Applying Lemma~\ref{claim:overlap} once more, $|\ip{\Gamma | w}| \geq 1- 2(11/48 + 1/48) = 1/2$. Finally, observe that the left Schmidt vectors of $\ket{w}$ are identical to the left Schmidt vectors of the state obtained from applying the successive trimming procedure to $\ket{u} = \ls(u')$ instead.  
\end{proof}

\subsection{Error Reduction}

The final step of the algorithm consists in reducing the error of the viable set, transforming the $(i, p_1(n),p_2(n), \frac{1}{2})$-viable set produced in the previous step to a $(i, p(n)p_1(n), p(n)p_2(n), \frac{\ce}{n})$-viable set, where $p(n)$ is a fixed polynomial. The key ingredient is the construction of an operator $K$ that to very good approximation keeps the ground state $\ket{\Gamma}$ fixed while cutting the norm of all vectors orthogonal to it by an inverse polynomial factor; this construction is detailed in Section~\ref{ss:4.1}. The structure of $K$ will be such that it can be decomposed as a sum of polynomially many terms of the form $ A_j \otimes B_j$ across the $(i,i+1)$ cut, with each $A_j$ of only polynomial Schmidt rank across every cut (see Corollary \ref{c:agsp}).  Though the application of $K$ moves complete vectors closer to the ground state, in our case where only a partial description of the ``left half'' of the vector is known (the viable set), it is this decomposition that allows us to generate a new viable set guaranteed to contain a good witness: the set will be obtained as the result of applying each $A_j$ to the elements of the viable set coming from the previous step. 

We summarize this last step of the algorithm:

\begin{topbotframe}
\begin{quote} {\bf Algorithm step 4: error reduction}\\
 Let $S_{i}^{(3)}$ be the set constructed as a result of the size and bond trimming steps described in the previous section.
\begin{enumerate}
\item[4.] Randomly select a sampling AGSP $K$ (Definition \ref{d:sagsp}) with $m$ and $l$ as in Corollary \ref{c:agsp}, and $q(n)=1/n$.  Write $K= \sum_j A_j \otimes B_j$ as in~\eqref{e:decomp}.  Return $S_{i}^{(4)}:= \{ A_j \ket{s}:  \ \ket{s}\in S_{i}^{(3)}\}$.
\end{enumerate}
\end{quote}
\end{topbotframe}
\newpage

That this step can be carried out in polynomial time follows from the properties of $K$ as detailed in Section~\ref{ss:4.1} below. The following claim shows that the step results in a viable set with the desired parameters. 

\begin{claim}  \label{c:errorreduction} For any choice of polynomial $q(n)$, there exists a polynomial $p(n)$ such that Step 4. described above maps any given $(i,s,b, \frac{1}{2})$-viable set $S$ into a $(i,p(n)s, p(n)b, \ce/q(n))$-viable set $S'$ with success probability $1- 1/n^3$.  
\end{claim}

\begin{proof}
Aplying Corollary \ref{c:agsp} with the polynomial $q(n)/2$, we obtain that for a proper choice of the parameters $m$ and $\ell$, and  with probability at least $1- 1/n^3$,  the sampling AGSP $K$ will have the desired properties.  Note that only $P_i$ acts across the boundary cut $(i,i+1)$ and that we can decompose it as the sum of $d^2$ terms as $P_i= \sum_{j,k=1}^d E_j\otimes F_k$.  Since furthermore $P_i$ appears no more than $\kappa \log (n)$ times in each term  in $K$ (for some $\kappa =O(\frac{1}{\eps})$), using the decomposition of $P_i$ within each term $P_I$ of $K$ we can write 
\begin{equation} \label{e:decomp}
 K= \sum_{j=1}^{d^{2\kappa \log n}\ell} A_j\otimes B_j 
 \end{equation}
 as the sum of polynomially many terms with $A_j$ acting only to the left of the cut and $B_j$ acting to the right.  Define $L:= \{A_j\}$ and set $S':= \{ A_j \ket{s}: A_j \in L, \ \ket{s}\in S\}$.  
Letting $\ket{v} = \sum \lambda_j \ket{a_j} \ket{b_j}$ be a witness for $S$, (\ref{e:cutdown}) yields $|\ip{\Gamma | \frac{Av}{\|Av\|}} |\geq 1- \ce/(2q(n))$ .  Given that $\|K-A\| \leq \ce/(2q(n))$, we get $|\ip{ \Gamma|\frac{Kv}{\norm{Kv}}} |\geq 1 - \ce/q(n)$, so that $\frac{Kv}{\|Kv\|}$ is a witness for $S'$ achieving the claimed error.  

Each step of the construction: generating the randomness needed for $K$, the computation of the $A_j$ and the construction of $S'$ can be done efficiently since there are only a polynomial number of terms involved, and the matrix product operator representations of the $A_j$ have polynomial size bond dimension and can be efficiently computed. As a result, the set $S'$ has size a fixed polynomial times that of $S$, and the bond dimension of vectors in $S'$ is also a fixed polynomial times the bond dimension of vectors in $S$. 
\end{proof}

\subsubsection{The sampling AGSP}\label{sec:agsp} \label{ss:4.1}

In~\cite{AradLV} the notion of an \emph{approximate ground state projection} (AGSP) associated to a local Hamiltonian $H$ was introduced. Intuitively, an AGSP is a local operator which, when applied to an arbitrary state $\ket{\Psi}$, moves it closer to the ground state $\ket{\Gamma}$. Elaborate (and very efficient) constructions of AGSPs were given in~\cite{AradKLV12area}. However, these constructions do not keep track of the effect of the AGSP on cuts to the left of the cut of interest, and it is thus unclear if they can be used as part of an algorithmically efficient procedure. Instead, here we introduce a less efficient but simpler construction that will fit our purpose.  

Our starting point is an operator $A$ that approximates the projection onto the ground state $\Gamma$, defined as $$A\,: =\,\Big( \frac{1}{1- \frac{\eps_0}{n}} \Big( 1- \frac{H}{n}\Big)\Big)^{m}.$$
 The operator $K$ is then formed from a polynomial sample of the exponentially many terms obained by expanding the $m$-th power in the definition of $A$. Using a matrix-valued Chernoff bound (Theorem \ref{t:mcb}), this polynomial sample can be shown to provide a close approximation to $A$ with high probability (Corollary \ref{c:agsp}). We now proceed with the details. 

First note that $A$ is positive semidefinite, has operator norm $1$ and satisfies $A \ket{\Gamma} = \ket{\Gamma}$. Furthermore, for any polynomial $q(n)$, fixing $m= \Theta(\frac{1}{\eps}n \log \frac{q(n)}{\ce})$ we have that for any $\ket{\Gamma^{\perp}}$orthogonal to the ground state $\ket{\Gamma}$ and such that  $\norm{\ket{\Gamma^{\perp}}} =1$,
\begin{equation} \label{e:cutdown} |\ip{ \Gamma^{\perp}|A|\Gamma^{\perp}} |\leq \Big(\frac{1- \eps_1/n}{1- \eps_0/n}\Big) ^{m}  = O\Big( \big(1- \frac{\eps}{n}\big)^ {\Omega(\frac{1}{\eps}n \log \frac{q(n)}{\ce})}\Big) \leq \frac{\ce}{2q(n)}. 
\end{equation}

Write $P_i:= (1- H_i)$, $C:=\frac{1}{1-\eps_0/n}$.  For any integer $m$, expand
\begin{equation}\label{eq:a-exp}
 A= C^m\frac{1}{n^m} \sum_{I=(i_1,\ldots,i_m)\in \{1,\ldots,n\}^m} P_I, \qquad\mbox{ with }\, P_I \,:=\, \prod_{j=1}^m P_{i_j}.
\end{equation}

\begin{definition} \label{d:sagsp} Define a {\it sampling AGSP} operator $K :=C^m\frac{1}{\ell} \sum_{j=1}^{\ell} P_{I_j}$ to be the average of $\ell$ terms $P_{I_j}$ chosen uniformly at random from all terms in the expansion~\eqref{eq:a-exp} of $A$. 
\end{definition}

We note that we may not be given $\eps_0$ and therefore cannot specify the constant $C$ explicitly. However we observe that any multiple of $K$ will suffice for use within the algorithm; only the resulting vectors need to be normalized. 

The following is implied by the variant of the Matrix-valued Chernoff bound~\cite[Theorem 1.6] {Tropp11chernoff}.

\begin{theorem}[Matrix-valued Chernoff bound] \label{t:mcb}
Let $X_i$ be $d\times d$ i.i.d. matrix random variables such that $\mathrm{E}[X_i]=X$, $\norm{X_i-X} \leq R$, and $\sigma^2 := max \{ E((X_i -X) (X_i -X)^*), E((X_i -X)^* (X_i -X)) \}$.  Then for all integers $\ell$ and $t\geq 0$, 

\begin{equation} \label{e:mcb}
\Pr\Big( \Big\|\frac{1}{\ell}\sum_{k=1}^\ell X_i - X \Big\| \geq t \Big) \leq 2d e^{-\frac{\ell t^2}{2\sigma^2 +2Rt/3}}.
\end{equation}
\end{theorem}

The following is an immediate corollary.

\begin{corollary}\label{claim:sampling-agsp} \label{c:agsp}
For any polynomial $q(n)$, there exists  $m= O(\frac{1}{\eps} n \log \frac{q(n)}{\ce})$ and $\ell= n^{O(1/\eps) }$ (where the implied constants may depend on the degree of $q(n)$) such that with probability at least $1- 1/n^3$ the sampling AGSP operator $K$ from Definition \ref{d:sagsp} has the following properties:
\begin{enumerate}
\item  $\|K-A\|\leq\frac{1}{q(n)}$,
\item Every projection $P_j$ appears no more than $\kappa \log (n)$ times in any term $P_I$ of $K$ for some $\kappa = O (\frac{1}{\eps})$.
\end{enumerate}
\end{corollary}

\begin{proof}  We apply Theorem \ref{t:mcb} with $X_i = C^m P_{I_i}$, $X=A$ and $t= \frac{1}{q(n)}$, noting the bounds  $R= C^m +1$, $\sigma^2 \leq (C^m +1)^{2}$.  Using these in (\ref{e:mcb}) yields 1. with probability at least
\[1-2d \exp ({-\frac{\ell}{4 C^{2m} q(n)^2}}).\] 
 Choosing $m$ as prescribed and noting that $C^{2m} = n^{O(\frac{1}{\eps})}$ (where the constant in the exponent may depend on the degree of $q$) leads to the probability of failure bounded by $ \exp ( - \ell/(n^{O(1/\eps)}q(n)^2))$ which for the specified choice of $\ell$  is exponentially small and in particular can be made less than $n^{-3}/2$ with an appropriate choice of constants. 

Property 2 follows from elementary probability: letting $Y_j$ to be a random variable counting the number of times $P_j$ appears in a given term, $Y_j$ has mean $\frac{m}{n}$ and variance bounded by $\frac{m}{n}$ and thus
\[\Pr\Big( \Big|Y_j - \frac{m}{n}\Big| \geq a \sqrt{\frac{m}{n}}\Big) \,\leq\, e^{-\Omega(a^2)}. \]
 For a proper choice of $a=O(\sqrt{\log (n \ell m)})$ the probability is bounded by $\frac{1}{2n^3 \ell m}$. By a union bound, the probability of every projection $P_j$ appearing no more than $(a \sqrt{\frac{m}{n}} + \frac{m}{n})$ times in any term of $K$ is bounded below by $1 - \frac{1}{2n^3}$.  With the prescribed choices of $a$ and $m$,  $a \sqrt{\frac{m}{n}} + \frac{m}{n}$ is upper bounded by $O((1/\eps)\log n)$.
\end{proof}

\subsection{Proof of Theorem~\ref{thm:main}}  

With the four steps (extension, cardinality reduction, bond reduction, and error reduction) of the algorithm established, the proof of our main theorem follows without difficulty. 

\begin{proof}[Proof of Theorem~\ref{thm:main}]
Claims~\ref{c:ext},~\ref{c:2},~\ref{c:trim},~\ref{c:errorreduction} together show that the succession of the four steps of the algorithm detailed in the previous sections transforms any $(i-1, p(n)p_1(n), p(n) p_2(n), \ce/n)$-viable set $S_{i-1}$ into an $(i, p(n)p_1(n), p(n)p_2(n), \ce/n)$-viable set $S_{i}$, where $p,p_1$ and $p_2$ are all fixed polynomials independent of $i$.  Moreover, this transformation can be executed in probabilistic polynomial time, with a success probability at least $1- \frac{1}{n^3}$. The effect of each step on the parameters of the viable step is summarized in the following table: 
\begin{center}
\begin{tabular}{llllll}
&&$i$ & $s$ & $B$ & $\delta$   \\ \hline
Start & & $i-1$ & $p(n)p_1(n)$ &$p(n)p_2(n)$& $\ce/n$ \\
&&&&&\\
{\bf Extension:} & $\rightarrow$ & ${\crd i} $ & $ {dp(n)p_1(n)}$ &$p(n)p_2(n)$ & $\ce/n$ \\
{\bf Size Trimming:} & $\rightarrow$ & $i$ & ${\color{red} p_1(n)}$ &  $p'(n)p_2(n)$&  ${ 1/12}$  \\
 {\bf Bond Trimming: }    & $\rightarrow$& $i$ & $p_1(n)$& ${\color{red} p_2(n)}$& ${ 1/2}$ \\
  {\bf Error reduction:} & $\rightarrow$&  $ i$ & $p(n) p_1(n)$& $ p(n) p_2(n)$&  ${\color{red} \ce/n }$\\ 
\end{tabular}
\end{center}

Starting from the set $\{1\}$, which is trivially a $(0,\ce/n)$-viable set, and proceeding inductively we efficiently construct a $(n, p(n)p_1(n), p(n)p_2(n), \ce/n)$-viable set, with success probability at least $1- \frac{1}{n^2}$. From this viable set we show how to obtain an inverse polynomial approximation to the ground state.

For this we first observe that the error reduction step in the final iteration can be modified to produce a $(n, p'(n)p_1(n), p'(n) p_2(n), \ce/(np(n)))$-viable set $S$, for any fixed polynomial $p(n)$ of our choice; for this it suffices to set $q(n)=np(n)$ instead of $q(n)=n$ in this step. Note that given the index $i=n$, the condition that $S$ is $(n,\ce/(np(n))$-viable simply means that there is a $\ce/(np(n))$ approximation to the ground state supported on $S$. Such an approximation has energy at most $\eps_0 + \frac{1}{p(n)}$, and can be found by solving the convex program
\begin{align*}
\mathrm{min}&\quad \sum_{j=1}^{n-1}\, \tr(H_j \,\sigma )\\
&\quad \tr(\sigma) = 1,\quad \sigma \geq 0,\notag
\end{align*} 
which is analogous to~\eqref{eq:conv} but for the omission of the constraint on the boundary contraction. By Lemma~\ref{l:leading}, the leading eigenvector $\ket{u}$ of an optimal solution $\sigma$ satisfies $| \ip{ u | \Gamma}| \geq 1- 1/p(n)$, as required. Moreover, $\sigma$ and $\ket{u}$ can be computed efficiently, as detailed in Section~\ref{sec:size-convex}.
\end{proof}

\section{Concluding Remarks}\label{sec:conclusion}


\begin{itemize}
\item We have made no attempt to pin down the exact polynomial running time of the algorithm since it can undoubtedly be optimized further.  In particular the convex optimization can be rewritten as a semi-definite program opening the door to applying  existing machinery involving matrix multiplicative weights and dimension reduction which may result in a fast combinatorial algorithm.
\item Though for the best known parameters, we must rely on the structural results (Lemmas \ref{l:mpsapprox} and \ref{constantbondapprox}) from~\cite{AradKLV12area}, we note that a polynomial bound on the running time could have been obtained using the earlier results of Hastings~\cite{Hastings07area}.  Furthermore it is worth noting the contrast of the AGSP used in the error reduction and those used to prove the area results in~\cite{AradKLV12area}.  The construction and analysis of the AGSP used here is far simpler than the one from~~\cite{AradKLV12area}, but it is too weak to imply an area law by itself.  On the other hand, the AGSP from~\cite{AradKLV12area} is tailored to a particular cut, and there is no control of its behavior across distant cuts; therefore it cannot be used for error reduction in our context.

One of the lessons to be drawn is that an AGSP is a powerful tool for exploiting the structure of a gapped system.  To work with the spectral gap condition, one might solely rely on the area law bound on entanglement and the resulting bounds on the bond dimension of the MPS of the ground state.  As the results of Schuch and Cirac~\cite{SchuchCV08hard} show under these conditions alone the problem of finding the ground state remains NP-hard.  The AGSP provides the further structure that makes a polynomial time algorithm possible.  It seems a natural and fruitful endeavor to investigate the role of AGSPs in the study of 2D systems.  
\item  An interesting open question is whether there exists a provably efficient \emph{local} algorithm, such as DMRG, for solving the problem considered here.  We note that DMRG itself is known to fail for certain instances, as one can construct (admittedly highly contrived) choices of initial conditions for the algorithm under which finding the global optimum requires solving an NP-hard problem~\cite{Eisert06hard}.

 \end{itemize}

\bibliographystyle{ieeetr}
\bibliography{algo}

\begin{thebibliography}{10}

\bibitem{Feynman82simulating}
R.~P. Feynman, ``Simulating physics with computers,'' {\em International
  journal of theoretical physics}, vol.~21, no.~6, pp.~467--488, 1982.

\bibitem{AharonovGIK091d}
D.~Aharonov, D.~Gottesman, S.~Irani, and J.~Kempe, ``The power of quantum
  systems on a line,'' {\em Communications in Mathematical Physics}, vol.~287,
  no.~1, pp.~41--65, 2009.

\bibitem{GottesmanI13translation}
D.~Gottesman and S.~Irani, ``The quantum and classical complexity of
  translationally invariant tiling and hamiltonian problems,'' {\em Theory of
  Computing}, vol.~9, no.~2, pp.~31--116, 2013.

\bibitem{White92dmrg}
S.~R. White, ``Density matrix formulation for quantum renormalization groups,''
  {\em Phys. Rev. Lett.}, vol.~69, pp.~2863--2866, 1992.

\bibitem{SchuchC10algo}
N.~Schuch and J.~I. Cirac, ``Matrix product state and mean-field solutions for
  one-dimensional systems can be found efficiently,'' {\em Phys. Rev. A},
  vol.~82, p.~012314, 2010.

\bibitem{AradKLV12area}
I.~Arad, A.~Kitaev, Z.~Landau, and U.~Vazirani, ``An area law and
  sub-exponential algorithm for {1D} systems,'' in {\em Proceedings of the 4th
  Innovations in Theoretical Computer Science (ITCS)}, 2013.

\bibitem{AharonovAI10algo}
D.~Aharonov, I.~Arad, and S.~Irani, ``Efficient algorithm for approximating
  one-dimensional ground states,'' {\em Phys. Rev. A}, vol.~82, p.~012315,
  2010.

\bibitem{Vidal03mps}
G.~Vidal, ``Efficient classical simulation of slightly entangled quantum
  computations,'' {\em Phys. Rev. Lett.}, vol.~91, p.~147902, 2003.

\bibitem{Hastings07area}
M.~B. Hastings, ``An area law for one-dimensional quantum systems,'' {\em
  Journal of Statistical Mechanics: Theory and Experiment}, vol.~2007, no.~08,
  p.~P08024, 2007.

\bibitem{Perez07MPS}
D.~Perez-Garcia, F.~Verstraete, M.~M. Wolf, and J.~I. Cirac, ``Matrix product
  state representations,'' {\em Quantum Info. Comput.}, vol.~7, no.~5,
  pp.~401--430, 2007.

\bibitem{AradLV}
I.~Arad, Z.~Landau, and U.~Vazirani, ``Improved one-dimensional area law for
  frustration-free systems,'' {\em Phys. Rev. B}, vol.~85, p.~195145, 2012.

\bibitem{Tropp11chernoff}
J.~A. Tropp, ``User-friendly tail bounds for sums of random matrices,'' {\em
  Foundations of Computational Mathematics}, vol.~12, pp.~389--434, 2012.

\bibitem{SchuchCV08hard}
N.~Schuch, I.~Cirac, and F.~Verstraete, ``Computational difficulty of finding
  matrix product ground states,'' {\em Phys. Rev. Lett.}, vol.~100, p.~250501,
  2008.

\bibitem{Eisert06hard}
J.~Eisert, ``Computational difficulty of global variations in the density
  matrix renormalization group,'' {\em Phys. Rev. Lett.}, vol.~97, p.~260501,
  2006.

\end{thebibliography}

\end{document}